\documentclass{article}

\usepackage{PRIMEarxiv}

\usepackage[utf8]{inputenc} 
\usepackage[T1]{fontenc}    
\usepackage{hyperref}       
\usepackage{url}            
\usepackage{booktabs}       
\usepackage{amsfonts}       
\usepackage{amsmath,amssymb}
\usepackage{mathtools}
\usepackage{nicefrac}       
\usepackage{microtype}      
\usepackage{lipsum}
\usepackage{fancyhdr}       
\usepackage{graphicx}       

\usepackage{times}
\usepackage{subfig}
\usepackage{xcolor}
\usepackage{array}
\usepackage{cite}
\usepackage{multirow}
\usepackage{adjustbox}
\usepackage[scr=dutchcal]{mathalfa}
\usepackage[font=small,labelfont=bf]{caption}
\usepackage{enumitem}
\usepackage{tikz}
\usetikzlibrary{arrows,matrix} 

\DeclareMathOperator{\diag}{diag}

\newtheorem{lemma}{Lemma}

\newtheorem{assumption}{Assumption}

\newtheorem{remark}{Remark}

\newtheorem{proof}{Proof}

\usepackage{tikz-network}
\tikzset{
    block/.style = {draw, fill=white, rectangle, minimum height=.75cm, minimum width=1.25cm},
    Circle/.style = {draw, thick, fill=white, circle, minimum height=.75cm, minimum width=1.25cm},
    tmp/.style  = {coordinate}, 
    sum/.style= {draw=white, fill=white, circle, minimum height=.75cm, minimum width=1.25cm},
    input/.style = {coordinate},
    output/.style= {coordinate},
    pinstyle/.style = {pin edge={to-,thin,black}
    }
}

\title{
Distributed Nonlinear Control of Networked Two-Wheeled Robots under Adversarial Interactions
}

\author{
  Moh Kamalul Wafi\\
  Department of Engineering Physics \\
  Institut Teknologi Sepuluh Nopember (ITS), Surabaya, Indonesia\\
  \texttt{\{kamalul.wafi\}@its.ac.id} \\
   \And
  Ahmad Ataka \\
  Department of Electrical and Information Engineering \\
  Universitas Gadjah Mada (UGM), Yogyakarta, Indonesia\\
  \AND
  Yul Yunazwin Nazaruddin \\
  Department of Engineering Physics \\
  Institut Teknologi Bandung (ITB), Bandung, Indonesia \\
  \And
  Bayu Jayawardhana \\
  Engineering and Technology Institute Groningen, Faculty of Science and Engineering \\
  University of Groningen, the Netherlands \\
}

\begin{document}
\maketitle

\begin{abstract}
This paper studies distributed trajectory tracking for networks of nonholonomic mobile robots under adversarial information exchange.
An exact global input--output feedback linearization scheme is developed to regulate planar position outputs, yielding linear error dynamics without prescribing internal state trajectories.
To mitigate corrupted neighbor information, a resilient desired-signal construction is proposed that combines local redundancy with trusted in-neighbor signals, without requiring adversary detection or isolation.
When sufficient redundancy is available, the method suppresses adversarial influence and recovers nominal tracking performance.
If redundancy conditions are violated, adversarial effects enter as bounded disturbances and the tracking error remains ultimately bounded.
Simulation results on star, cyclic, and path topologies validate the analysis and demonstrate the superior resilience of cyclic networks due to distributed information propagation.
\end{abstract}
\allowdisplaybreaks


\section{Introduction}

Distributed coordination and trajectory tracking of networked robotic systems have received significant attention due to their applications, including autonomous transportation, surveillance, and cooperative manipulation \cite{ref1,ref2,Wafi-QuadrupleTank,ref17}.
In such systems, agents rely on local information exchange over a communication graph to collectively follow a desired trajectory or formation, often specified by a leader or navigator. While distributed control offers scalability and robustness to single-point failures, their performance can be severely degraded by unreliable, faulty, or malicious information exchanged over the network \cite{ref4,ref5}.

Most existing distributed control strategies assume benign communication environments and focus on linear agent models or full-state feedback \cite{ref6,Wafi-ThreeTank}.
However, many robotic platforms, such as wheeled mobile robots, are governed by nonlinear and nonholonomic dynamics, rendering linear consensus-based methods inadequate. Moreover, adversarial interactions arising from sensor faults, spoofed transmissions, or malicious agents can corrupt shared information and propagate through the network, leading to instability or large tracking errors \cite{ref8,ref9,ref10}.
These issues are particularly pronounced for two-wheeled mobile robots, where coupling between translational and rotational motion amplifies the effect of corrupted signals.
The challenge is further exacerbated when higher-order derivatives of network signals are required, as is typical in feedback linearization frameworks for nonlinear systems.

Recent works on resilient consensus and secure multi-agent coordination have addressed adversarial behaviors using graph-theoretic redundancy, trimming, or filtering techniques \cite{ref11,ref12,ref16}. These approaches exploit local redundancy to suppress extreme or inconsistent neighbor information without explicit adversary identification.
Nevertheless, most of the existing results are developed for single- or double-integrator models and do not directly extend to nonlinear systems requiring exact input–output linearization.
Furthermore, several methods rely on centralized filtering, adversary detection, or static assumptions on trusted agents \cite{ref13,ref14}, which conflict with the goal of fully distributed implementation and increase vulnerability to coordinated attacks.
Addressing resilience for nonlinear, nonholonomic multi-agent systems under adversarial information exchange therefore remains an open and practically relevant problem \cite{ref15}. Other secure multi-agent coordination methods use homomorphic encryption techniques combined with distributed control approaches to ensure secure computation of distributed control through third-party cloud services \cite{matteo2023}.  

In this paper, we propose a distributed nonlinear control framework for trajectory tracking of nonholonomic robotic networks under adversarial information exchange.
The approach regulates only the planar position outputs via global input–output feedback linearization, resulting in exact linear error dynamics while avoiding explicit control of internal states.
To handle adversarial interactions, we introduce a resilient construction of the local desired signal that combines trusted navigator information with time-varying trimming of neighbor signals.
The trimming mechanism operates in a fully distributed manner and adapts over time, allowing different neighbors to be discarded at different instants.

The proposed method does not require adversary detection or isolation and remains compatible with arbitrary directed topologies provided the navigator is reachable.
When sufficient local redundancy is available, adversarial effects are suppressed at the reference level; otherwise, the closed-loop system admits provable worst-case robustness bounds.
Extensive simulations on star, cyclic, and path networks demonstrate that, while all topologies achieve nominal tracking in the absence of adversaries, cyclic graphs exhibit superior resilience under adversarial conditions due to distributed information propagation.

The main contributions of this paper are as follows:
\begin{enumerate}[leftmargin=*]
    \item A distributed input–output feedback linearization framework for nonholonomic robotic networks that regulates only output variables and yields linear error dynamics.
    \item A resilient desired-signal construction that mitigates adversarial interactions using trusted in-neighbor access and time-varying trimming, without detection or isolation.
    \item A robustness analysis showing bounded tracking errors when graph redundancy is insufficient.
    \item Numerical simulations illustrating the impact of network topology on resilience and highlighting the advantages of cyclic information flow under adversarial interactions.
\end{enumerate}

Notation:
Let $\mathbb R$, $\mathbb R_{\ge0}$, and $\mathbb C$ denote the sets of real, nonnegative real, and complex numbers, respectively. The identity and zero matrices of dimension $n$ are denoted by $I_n$ and $0_n$. The vector of ones is denoted by $\mathbf 1_n$. For a matrix $A$, $A\succ0$ ($A\succeq0$) indicates that $A$ is positive definite (positive semidefinite). The Kronecker product is denoted by $\otimes$, and $\diag\{\cdot\}$ denotes a block-diagonal matrix. The Euclidean norm is written as $\|\cdot\|$.

\section{Communication Network}\label{sec:ComNetwork}

The cooperative motion architecture is organized around a collection of 
$m{+}1$ robotic vehicles (two-wheeled robots) indexed by $\mathcal{V}=\{0,1,\dots,m\}$. Vehicle $0$ acts as a \emph{navigator unit} that injects a desired motion signal, while vehicles $1,\dots,m$ operate as \emph{tracking units} that sense locally and exchange information with nearby vehicles to coordinate their motion.

Information exchange among the vehicles is described by a weighted digraph $\mathcal{G} = (\mathcal{V},\mathcal{E},\mathcal{W})$, where an edge $(i,j)\in\mathcal{E}$ with weight $w_{ij}> 0$ signifies that vehicle $i$ receives the information from vehicle $j$.
The in-neighbor set of vehicle $i$ is $\mathcal{N}_i=\{j\in\mathcal{V} \mid (i,j)\in\mathcal{E}\}$. To isolate vehicle--vehicle communication from navigator--vehicle influence, we introduce two induced subgraphs of $\mathcal{G}$ (see Fig.~\ref{Fig:network}):
\begin{itemize}
    \item $\mathcal{G}_m=(\mathcal{V}_m,\mathcal{E}_m,\mathcal{W}_m)$ with $\mathcal{V}_m=\{1,\dots,m\}$, capturing inter-vehicle interactions;
    \item $\mathcal{G}_0=(\mathcal{V}_0,\mathcal{E}_0,\mathcal{W}_0)$ with $\mathcal{V}_0=\{0\}\cup\{i:(i,0)\in\mathcal{E}\}$, capturing navigator-to-vehicle interaction.
\end{itemize}
\begin{figure}[h!]
    \centering
    \scalebox{1}{{\begin{tikzpicture}
            \centering
            \Text[x=1,y=1.3,fontsize=\small]{$\mathcal{G}$};
            \Vertex[x=-.375,y=-.375,label=$0$,color=blue,opacity=0.2,size=.5]{L}
            \Vertex[x=-1,y=1,label=$1$,color=red,opacity=0.2,size=.5]{1}
            \Vertex[x=-1,y=-1,label=$2$,color=red,opacity=0.2,size=.5]{2}
            \Vertex[x=1,y=-1,label=$3$,color=red,opacity=0.2,size=.5]{3}
            \Vertex[x=1,y=0,label=$4$,color=red,opacity=0.2,size=.5]{4}
            \Vertex[x=0,y=1,label=$5$,color=red,opacity=0.2,size=.5]{5}
            \Edge[Direct,color=blue,label=$w_{10}$](L)(1)
            \Edge[Direct,color=blue,label=$w_{40}$](L)(4)
            \Edge[Direct,color=blue,label=$w_{50}$](L)(5)
            \Edge[Direct,label=$w_{21}$](1)(2)
            \Edge[Direct,label=$w_{32}$](2)(3)
            \Edge[Direct,bend=30,label=$w_{45}$](5)(4)
            \Edge[Direct,bend=30,label=$w_{54}$](4)(5)
        \end{tikzpicture}}}
    \qquad
    \scalebox{1}{{\begin{tikzpicture}
            \centering
            \Text[x=1,y=1.3,fontsize=\small]{$\mathcal{G}_m$};
            \Vertex[x=-1,y=1,label=$1$,color=red,opacity=0.2,size=.5]{1}
            \Vertex[x=-1,y=-1,label=$2$,color=red,opacity=0.2,size=.5]{2}
            \Vertex[x=1,y=-1,label=$3$,color=red,opacity=0.2,size=.5]{3}
            \Vertex[x=1,y=0,label=$4$,color=red,opacity=0.2,size=.5]{4}
            \Vertex[x=0,y=1,label=$5$,color=red,opacity=0.2,size=.5]{5}
            \Edge[Direct,label=$w_{21}$](1)(2)
            \Edge[Direct,label=$w_{32}$](2)(3)
            \Edge[Direct,bend=30,label=$w_{45}$](5)(4)
            \Edge[Direct,bend=30,label=$w_{54}$](4)(5)
        \end{tikzpicture}}}
    \qquad
    \scalebox{1}{{\begin{tikzpicture}
            \centering
            \Text[x=1,y=1.3,fontsize=\small]{$\mathcal{G}_0$};
            \Vertex[x=-.375,y=-.375,label=$0$,color=blue,opacity=0.2,size=.5]{L}
            \Vertex[x=-1,y=-1,style={fill=white,opacity=0},size=.5]{2}
            \Vertex[x=-1,y=1,label=$1$,color=red,opacity=0.2,size=.5]{1}
            \Vertex[x=1,y=0,label=$4$,color=red,opacity=0.2,size=.5]{4}
            \Vertex[x=0,y=1,label=$5$,color=red,opacity=0.2,size=.5]{5}
            \Edge[Direct,color=blue,label=$w_{10}$](L)(1)
            \Edge[Direct,color=blue,label=$w_{40}$](L)(4)
            \Edge[Direct,color=blue,label=$w_{50}$](L)(5)
        \end{tikzpicture}}}
    \caption{Example of a graph $\mathcal{G}$ with $m=5$, vehicle-to-vehicle subgraph $\mathcal{G}_m$, and navigator-to-vehicle subgraph $\mathcal{G}_0$, showing the decoupling and assignment of $w_{ij}$.}
    \label{Fig:network}
\end{figure}
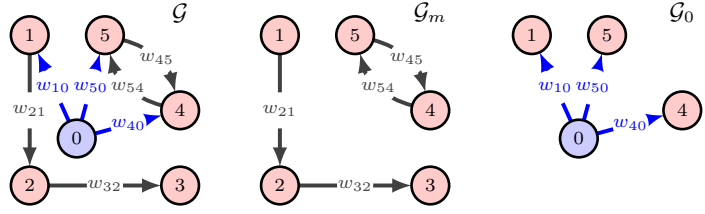

Regarding the subgraph $\mathcal{G}_m$, the in-degree and adjacency matrices are defined as $\mathbb{D}_m=\diag\{d_1,\dots,d_m\}$ where $d_i=\sum_{j:(i,j)\in\mathcal{E}_m}w_{ij}$, and $[\mathbb{A}_m]_{ij}=w_{ij}$. The corresponding Laplacian is $\mathbb{L}_m \coloneqq \mathbb{D}_m - \mathbb{A}_m$. Within $\mathcal{G}_0$, the navigator influences the tracking vehicles through the diagonal matrix $\mathbb{A}_0=\diag\{w_{10},\dots,w_{m0}\}$.
Aggregating both contributions yields the augmented Laplacian $\mathbb{L} = \mathbb{L}_m + \mathbb{A}_0$. 

Introducing the diagonal matrix $\mathbb{W}=\diag\{w_1,\dots,w_m\}$ where $w_i=d_i+w_{i0}$ and $w_i=1$ for all $i$, gives
\begin{equation*}
    \mathbb{L} \coloneqq \mathbb{L}_m + \mathbb{A}_0 = \mathbb{W} - \mathbb{A}_m.
\end{equation*}
For cooperative motion control, we adopt a balanced weighting design such that
\begin{equation}\label{eq:ComNet:balance}
    (\mathbb{L}-\mathbb{A}_0)\mathbf{1}_m=\mathbf{0}_m
    \quad\Leftrightarrow\quad
    (\mathbb{A}_m+\mathbb{A}_0)\mathbf{1}_m=\mathbf{1}_m.
\end{equation}
If $\mathbb{W}\neq I_m$, a normalized set of weights $\tilde w_{ij}=w_{ij}/w_i$, $\tilde w_{i0}=w_{i0}/w_i$ reproduces the balance property and yields the normalized Laplacian
$\tilde{\mathbb{L}}=\tilde{\mathbb{L}}_m+\tilde{\mathbb{A}}_0$, satisfying \eqref{eq:ComNet:balance}.

\begin{remark}[Navigator reachability]\label{rem:reachability}
If at least one tracking vehicle satisfies $w_{i0}>0$ and every vehicle is reachable from vehicle~$0$ along a directed path, then $\mathbb{L}=\mathbb{L}_m+\mathbb{A}_0$ 
is positive stable, i.e., all eigenvalues of $\mathbb L$ have positive real parts.
\end{remark}

\section{Problem Formulation and Preliminaries}

\subsection{Nominal Problem Formulation}

We adopt a general systems-theoretic formulation of nonlinear dynamics defined by the quadruple $\mathbb{S} \coloneqq (X, P, U, \phi)$,
where $X$ is a topological state space, $P$ is a strongly ordered group indexing time, $U$ is a linear space of control inputs, and
$\phi : P \times X \times U \to X$ is a transition map.
Throughout this work, we consider continuous-time systems with $P=\mathbb{R}$ and assume that $\phi$ is continuous.
Accordingly, state trajectories are mappings $\chi : P \to X$, where $\chi(t)$ denotes the system state at time $t \in P$.

We now specialize the formulation to a networked setting consisting of $m{+}1$ vehicles mentioned in Section~\ref{sec:ComNetwork}, \textit{without} adversarial interactions. Each vehicle $i\in\{0,1,\dots,m\}$ is governed by nonlinear dynamics
\begin{equation}\label{eq:Problem:dynamics}
    \dot{\chi}_i = \pi(\chi_i,\nu_i), \qquad y_i = h(\chi_i)
\end{equation}
where $\chi_i\in\mathbb{R}^n$ denotes the state vector, $\nu_i\in\mathbb{R}^p$ is the control input and $y_i\in\mathbb{R}^q$ is the output. The vector field $\pi:\mathbb{R}^n\times\mathbb{R}^p\to\mathbb{R}^n$ is assumed to be sufficiently smooth and identical across all vehicles and $h:\mathbb{R}^n\to\mathbb{R}^q$ is a smooth output map.
Vehicle $i$ forms a local desired signal $z_i\in\mathbb{R}^q$ as a weighted sum of its neighbors' outputs (including the navigator output $y_0\in\mathbb{R}^q$ if $0\in\mathcal{N}_i$):
\begin{equation}\label{eq:Problem:loc_ref}
    z_i \coloneqq [z_{i,1},\dots,z_{i,q}]^\top = \sum\nolimits_{j\in\mathcal{N}_i} w_{ij} y_j,
\end{equation}
with $w_{ij}> 0$. The local tracking error is
\begin{equation}\label{eq:Problem:loc_error}
    \epsilon_i \coloneqq y_i - z_i \in \mathbb{R}^q.
\end{equation}
To accommodate feedback laws that depend on higher-order output-error coordinates,
we define an augmented error state $\eta_i\in\mathbb{R}^r$ constructed from $\epsilon_i$
and locally available signals defined later.
The nominal (non-adversarial) control objective is to design control laws
$\nu_i = K_i(\eta_i)$, $K_i:\mathbb{R}^r\to\mathbb{R}^p$,
such that all tracking vehicles asymptotically follow the navigator,
\begin{equation}
    \lim_{t\to\infty} \|y_i(t)-y_0(t)\| = 0,
    \qquad \forall\, i\in\{1,\dots,m\}.
\end{equation}

\subsection{Non-nominal Problem Formulation}\label{subsec:Adversarial}

The communication network introduced in Section~\ref{sec:ComNetwork} may contain adversarial interactions that inject corrupted information into the distributed coordination process. Such interactions may arise from compromised vehicles, malicious transmissions, or sensor spoofing.

Let $\mathcal E_a \subseteq \mathcal E_m$ denote the set of adversarial interactions and define
$\mathcal N_i^a\coloneqq \{j\in\mathcal V \mid (i,j)\in\mathcal E_a\}$
as the set of adversarial neighbors of vehicle~$i$. Accordingly, the corrupted desired signal received by vehicle~$i$ is
\begin{equation}\label{eq:Problem:Adv_signal}
    z_i^a
    \coloneqq
    \sum\nolimits_{j\in\mathcal N_i^-} w_{ij} y_j
    +
    \sum\nolimits_{j\in\mathcal N_i^a} w_{ij} y_j^a,
\end{equation}
where $y_j^a(t)\coloneqq y_j(t)+a_{ij}(t)$ and
$\mathcal N_i^- \coloneqq \mathcal N_i\setminus\mathcal N_i^a$ denotes the set of non-adversarial neighbors. We assume that the adversarial perturbations are sufficiently smooth and satisfy
$\|a_{ij}^{(k)}(t)\|\le \bar a$ for $k=0,1,2,3$ and all $(i,j)\in\mathcal E_a$.

\begin{assumption}
For each vehicle $i$, the set of non-adversarial neighbors $\mathcal N_i^-$ is nonempty and vehicle $i$ is reachable from the navigator vehicle~$0$ through non-adversarial communication links.
\end{assumption}

\begin{remark}
Equation~\eqref{eq:Problem:Adv_signal} does not require adversary detection or isolation, and instead models adversarial interactions as bounded perturbations affecting distributed coordination.
\end{remark}

In the presence of adversarial interactions, the objective is to design distributed state-feedback control laws such that the tracking error remains bounded despite corrupted information exchange. In particular, exact asymptotic tracking,
\begin{equation*}
    \lim_{t\to\infty} \|y_i(t)-y_0(t)\| = 0,
\end{equation*}
is achieved when adversarial effects are sufficiently weak and the communication network provides adequate redundancy. Otherwise, the objective is to ensure ultimate boundedness,
\begin{equation}
    \limsup_{t\to\infty} \|y_i(t)-y_0(t)\|
    \le \delta(\bar a),
\end{equation}
where $\delta(\bar a)$ depends on the adversarial magnitude.

\subsection{Motivation for Distributed Nonlinear Control}

Let $\chi_{i}^\ast\in\mathbb{R}^n$ denote a desired state associated with vehicle $i$, and define the state tracking error
\begin{equation}\label{eq:Problem:state_error}
    e_i \coloneqq \chi_i - \chi_{i}^\ast = \chi_i - \sum\nolimits_{j\in\mathcal{N}_i} w_{ij} \chi_j.
\end{equation}
where $\chi_0$ denotes the navigator state.
Differentiating \eqref{eq:Problem:state_error} along the trajectories of the nonlinear dynamics
$\dot{\chi}_i = \pi(\chi_i,\nu_i)$ yields
\begin{equation}\label{eq:Problem:error_dyn}
    \dot e_i
    = \pi(\chi_i,\nu_i) - \dot \chi_{i}^\ast
    \eqqcolon \mathcal{F}_i(\chi_i,\chi_{i}^\ast) + \mathcal{H}_i(\chi_i)\nu_i,
\end{equation}
which is nonlinear and state dependent. A common approach is to approximate \eqref{eq:Problem:error_dyn} locally around the trajectory
$\chi_i^\ast$ via first-order linearization, leading to $\dot e_i = A_i e_i + B_i \nu_i$, where $A_i$ and $B_i$ are the corresponding Jacobians. Controllers designed for locally linearized models guarantee stability only in a neighborhood of $\chi_i^\ast$, and performance may degrade for large initial errors or under strong network-induced coupling. 

In contrast, this work regulates only the output variables $y_i=h(\chi_i)$ rather than the full state $\chi_i$, defining tracking errors at the output level and constructing higher-order error coordinates through output differentiation. This approach enables global input–output feedback linearization without prescribing desired trajectories for the internal states, which are instead shaped implicitly by the closed-loop dynamics.

\section{Distributed Nonlinear Control of Nonholonomic Networks}
\begin{figure}[t!]
    \centering
    \includegraphics[width=.5\linewidth]{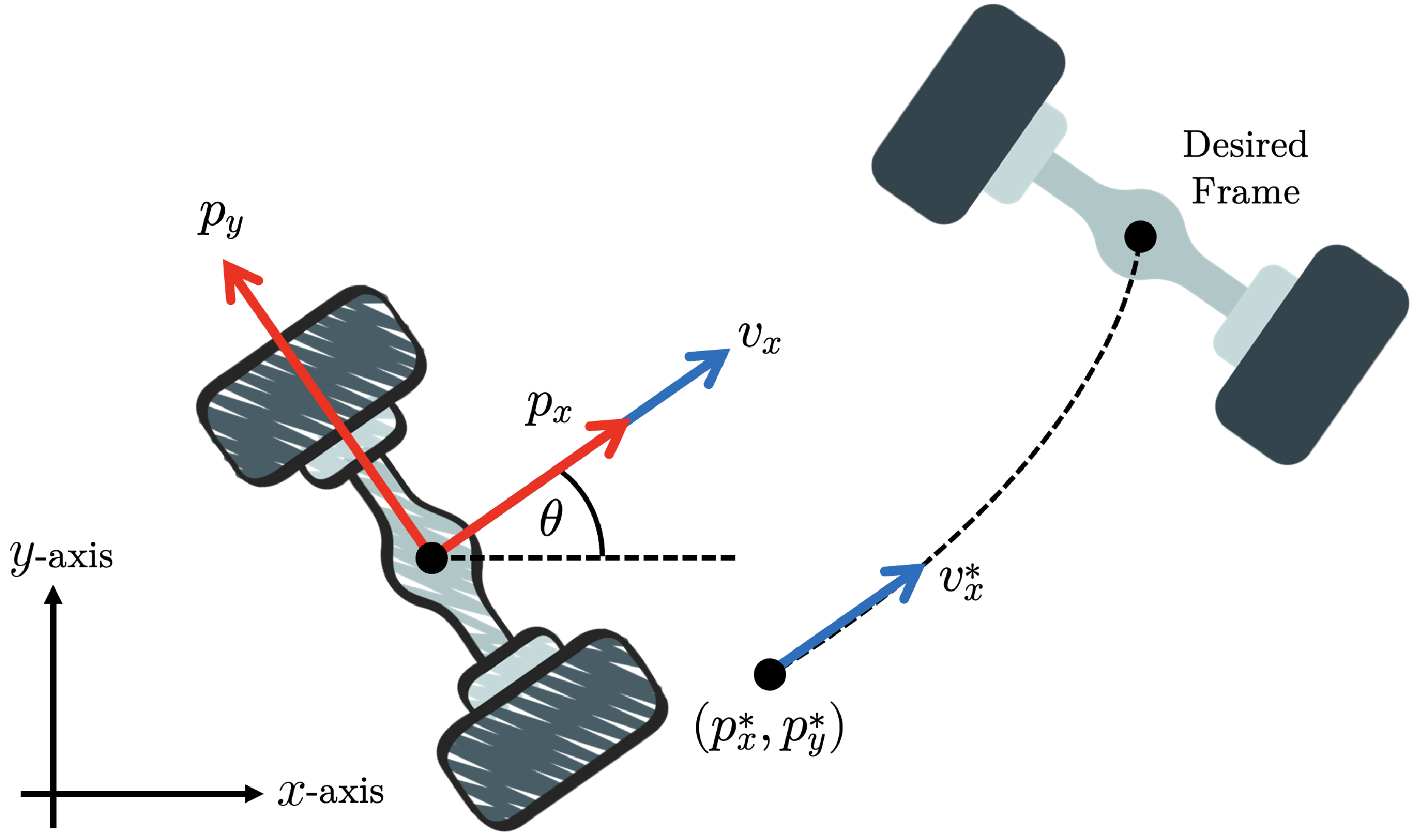}
    \caption{Nonlinear dynamics of a nonholonomic two-wheeled robot}
    \label{fig:placeholder}
\end{figure}
\subsection{Nonholonomic Network Model}

We consider a network of $m$ identical nonholonomic two-wheeled robotic vehicles governed in \eqref{eq:Problem:dynamics}, each described as
\begin{equation}\label{eq:Robots:nonlinear_dynamics}
    \left\{\begin{array}{lll}
         \dot p_{x,i} = v_{x,i}\cos\theta_i, & \dot \theta_i\phantom{_{\,x}} = \omega_i, & \dot \omega_i = T_i, \\
        \dot p_{y,i} = v_{x,i}\sin\theta_i, & \dot v_{x,i} = F_i, &
    \end{array}\right.
\end{equation}
where $\chi_i := [p_{x,i}, p_{y,i}, \theta_i, v_{x,i}, \omega_i]^\top\in\mathbb{R}^5$ is the state and $y_i = [p_{x,i}, p_{y,i}]^\top \in \mathbb{R}^2$ is the output.
Here, $p_{x,i}$ and $p_{y,i}$ denote the planar position of the vehicle center of mass, $\theta_i$ is the body orientation, $v_{x,i}$ is the longitudinal velocity along the body-fixed axis, and $\omega_i$ is the angular velocity. The physical control input is defined as $\nu_i \coloneqq [F_i, T_i]^\top \in \mathbb{R}^2$,
where $F_i$ and $T_i$ denote the longitudinal force and applied torque, respectively.

To enable global feedback linearization, we introduce a dynamic extension of the force input $\dot F_i = u_{i,1}$, treated as another system state, and define $T_i = u_{i,2}$. This results in $u_i \coloneqq [u_{i,1}, u_{i,2}]^\top$ where $u_{i,1},u_{i,2}\in\mathbb{R}$ are virtual control inputs. 
Therefore, the network dynamics admit the compact form
\begin{equation}\label{eq:Robots:dynamics}
    \dot{\bar x} = \mathbf{f}(\bar x,\bar u), \qquad \bar y = \mathbf{h}(\bar x),
\end{equation}
where $\bar x \coloneqq [x_1^\top,\dots,x_m^\top]^\top\in \mathbb{R}^{6m}$ represents the stacked state vector of the augmented state $x_i = [\chi_i^\top, F_i]^\top\in\mathbb{R}^{6}$ and $\bar u \coloneqq [u_1^\top,\dots,u_m^\top]^\top \in \mathbb{R}^{2m}$ denotes the stacked control input. 
The nonlinear vector field $\mathbf{f}:\mathbb{R}^{6m}\times\mathbb{R}^{2m}\to\mathbb{R}^{6m}$ is defined as $\mathbf{f}(\bar x,\bar u) \coloneqq [f(x_1,u_1)^\top, \dots, f(x_m,u_m)^\top]^\top$ and is assumed to be locally Lipschitz in $(\bar x,\bar u)$.

The output map is defined as $ \mathbf{h}(\bar x) = (I_m \otimes C)\bar x$
where $C\in\mathbb{R}^{2\times 6}$ selects the planar positions and the network output $\bar y = [y_1^\top,\dots,y_m^\top]^\top\in\mathbb{R}^{2m}$ collects the planar positions of all vehicles, with $y_i\coloneqq[p_{x,i},p_{y,i}]^\top$.
Finally, given the desired trajectory $y_0\coloneqq[p_{x,0},p_{y,0}]^\top$, then $\bar y_0 \coloneqq \mathbf{1}_m \otimes y_0 \in \mathbb{R}^{2m}$.

By stacking the local desired signals $z_i$ defined in \eqref{eq:Problem:loc_ref}, the network-level desired signal can be written as
\begin{equation}\label{eq:Robots:ref}
    \bar z \coloneqq [z_1^\top,\dots,z_m^\top]^\top =
    (\mathbb{A}_m\otimes I_2)\bar y
    + (\mathbb{A}_0\otimes I_2)\bar y_0.
\end{equation}
Accordingly, stacking the local tracking errors $\epsilon_i = y_i - z_i$ in \eqref{eq:Problem:loc_error} yields the network-level tracking error expression
\begin{equation}\label{eq:Robots:error}
    \bar\epsilon \coloneqq [\epsilon_1^\top,\dots,\epsilon_m^\top]^\top =
    (\mathbb{L}\otimes I_2)\bar y
    - (\mathbb{A}_0\otimes I_2)\bar y_0,
\end{equation}
or equivalently $\bar\epsilon = \bar y - \bar z$. This tracking-error formulation allows the closed-loop dynamics to be expressed explicitly in terms of output errors and their derivatives, which motivates the use of a global feedback linearization approach.

\subsection{Global Feedback Linearization}

We define a state transformation by augmenting the tracking error coordinates and their time derivatives induced by \eqref{eq:Robots:error}. For each vehicle $i$, define
\begin{equation*}
    \eta_i \coloneqq [\epsilon_i^\top,\dot\epsilon_i^\top,\ddot\epsilon_i^\top]^\top
    =
    [\epsilon_{x,i},
    \epsilon_{y,i},
    \dot\epsilon_{x,i},
    \dot\epsilon_{y,i},
    \ddot\epsilon_{x,i},
    \ddot\epsilon_{y,i}]^\top
    \in\mathbb{R}^{6},
\end{equation*}
and let
$\bar \eta = [\eta_1^\top,\dots,\eta_m^\top]^\top\in\mathbb{R}^{6m}$
denote the stacked tracking-error state.
Using \eqref{eq:Robots:nonlinear_dynamics}, with
$\epsilon_{x,i}=p_{x,i}-z_{i,1}$ and
$\epsilon_{y,i}=p_{y,i}-z_{i,2}$,
the first four rows of $\dot\eta_i$ follow directly from the definitions, while the highest-order derivatives satisfy
\begin{align*}
    \dddot\epsilon_{x,i}
    &=
    -\bigl(
        2F_i\omega_i\sin\theta_i
        +
        v_{x,i}\omega_i^2\cos\theta_i
    \bigr)
    -\dddot z_{i,1} +u_{i,1}\cos\theta_i-v_{x,i}u_{i,2}\sin\theta_i,
    \\
    \dddot\epsilon_{y,i}
    &=
    \phantom{-}\bigl(
        2F_i\omega_i\cos\theta_i
        -
        v_{x,i}\omega_i^2\sin\theta_i
    \bigr)
    -\dddot z_{i,2} +u_{i,1}\sin\theta_i+v_{x,i}u_{i,2}\cos\theta_i.
\end{align*}
Observe that the planar outputs
$y_i=[p_{x,i},p_{y,i}]^\top$
have relative degree three with respect to the virtual inputs $u_i$.
Indeed, the inputs $u_{i,1}$ and $u_{i,2}$ appear explicitly only after differentiating the output errors three times. Consequently, the transformed coordinates $\eta_i$ generate a Brunovsk\'y-type normal form consisting of two chains of integrators together with nonlinear coupling terms arising from the nonholonomic kinematics.

Hence, the stacked error dynamics can be written as
\begin{equation}\label{eq:Robots:error_dyn}
    \dot{\bar \eta}
    =
    \mathbf A\bar \eta
    +
    \mathbf B\bigl[
        \bar\lambda
        +
        \Psi\bar u
    \bigr],
\end{equation}
where
$\mathbf A\coloneqq I_m\otimes\mathcal A$
and
$\mathbf B\coloneqq I_m\otimes\mathcal B$,
with
\begin{equation}\label{eq:Robots:AB_matrix}
    \mathcal A =
    \begin{bmatrix}
        0 & 0 & 1 & 0 & 0 & 0 \\
        0 & 0 & 0 & 1 & 0 & 0 \\
        0 & 0 & 0 & 0 & 1 & 0 \\
        0 & 0 & 0 & 0 & 0 & 1 \\
        0 & 0 & 0 & 0 & 0 & 0 \\
        0 & 0 & 0 & 0 & 0 & 0
    \end{bmatrix},
    \quad
    \mathcal B =
    \begin{bmatrix}
        0 & 0 \\
        0 & 0 \\
        0 & 0 \\
        0 & 0 \\
        1 & 0 \\
        0 & 1
    \end{bmatrix}.
\end{equation}
Equation~\eqref{eq:Robots:error_dyn} represents an exact input--output linearization of the nonlinear network dynamics in the transformed coordinates.
The matrices $\mathcal A$ and $\mathcal B$ represent two decoupled chains of integrators associated with the planar tracking errors.
The nonlinear term $\bar\lambda$ is defined by
\begin{equation}\label{eq:Robots:lambda}
    \bar\lambda
    =
    \bar\sigma
    -
    \dddot{\bar z},
\end{equation}
where
$\bar\lambda=[\lambda_1^\top,\dots,\lambda_m^\top]^\top$ and 
$\bar\sigma=[\sigma_1^\top,\dots,\sigma_m^\top]^\top$.
For each vehicle $i$, 
$\lambda_i = \sigma_i - \dddot z_i$ and $\dddot z_i = [\dddot z_{i,1},\dddot z_{i,2}]^\top,$
with
\begin{equation*}
    \sigma_i
    =
    \begin{bmatrix}
        -2F_i\omega_i\sin\theta_i
        -
        v_{x,i}\omega_i^2\cos\theta_i
        \\
        \phantom{-}2F_i\omega_i\cos\theta_i
        -
        v_{x,i}\omega_i^2\sin\theta_i
    \end{bmatrix}.
\end{equation*}
The term $\dddot{\bar z}$ is assumed to be well-defined along admissible trajectories.
Finally, $\Psi = \diag\{\psi_1,\dots,\psi_m\},$ where
\begin{equation*}
    \psi_i
    =
    \begin{bmatrix}
        \cos\theta_i & -v_{x,i}\sin\theta_i \\
        \sin\theta_i & \phantom{-}v_{x,i}\cos\theta_i
    \end{bmatrix}.
\end{equation*}
Thus, $\Psi$ characterizes the coupling between the virtual inputs $\bar u$ and the highest-order tracking-error derivatives. Since $\det(\psi_i)=v_{x,i}$, the matrix $\Psi$ is invertible along trajectories satisfying
$v_{x,i}\neq0$.

\subsection{Distributed Nonlinear Control}

Using the feedback-linearized error dynamics \eqref{eq:Robots:error_dyn}, define the virtual control input
\begin{equation}\label{eq:Robots:control}
    \bar u
    =
    \Psi^{-1}
    \bigl[
        -\bar\lambda
        +
        \bar u^+
    \bigr],
\end{equation}
where $\bar u^+\in\mathbb R^{2m}$ is an auxiliary input.
Substituting \eqref{eq:Robots:control} into \eqref{eq:Robots:error_dyn} cancels the nonlinear terms and yields the exact linearized error dynamics
\begin{equation}\label{eq:Robots:linear_cl}
    \dot{\bar\eta}
    =
    \mathbf A\bar\eta
    +
    \mathbf B\bar u^+.
\end{equation}
Thus, the nonlinear tracking problem is transformed into the stabilization of the linear system \eqref{eq:Robots:linear_cl}.
We now choose $\bar u^+$ as the linear state-feedback law $\bar u^+ = \mathbf K\bar\eta,$
where $\mathbf K\in\mathbb R^{2m\times6m}$ is selected such that
$\mathbf A+\mathbf B\mathbf K$ is Hurwitz.
A convenient distributed choice is $\mathbf K = I_m\otimes\mathcal K,$
where $\mathcal K\in\mathbb R^{2\times6}$ satisfies
$\mathcal A+\mathcal B\mathcal K$
Hurwitz.
Such a matrix $\mathcal K$ always exists because the pair
$(\mathcal A,\mathcal B)$
is controllable. Indeed, $\mathbf A+\mathbf B\mathbf K = I_m\otimes(\mathcal A+\mathcal B\mathcal K),$
whose eigenvalues are precisely those of $\mathcal A+\mathcal B\mathcal K$ repeated $m$ times.

Although \eqref{eq:Robots:control} achieves exact cancellation, the term
$\dddot{\bar z}$
depends implicitly on neighboring outputs through \eqref{eq:Robots:ref}, and therefore depends on the control input itself.
To derive an implementable expression, observe from
\eqref{eq:Robots:error_dyn}
that
\begin{equation*}
    \dddot{\bar\epsilon}
    =
    \bar\sigma
    -
    \dddot{\bar z}
    +
    \Psi\bar u.
\end{equation*}
Since
$\bar\epsilon=\bar y-\bar z$, we also have $\dddot{\bar\epsilon} = \dddot{\bar y} - \dddot{\bar z},$
which gives $\dddot{\bar y} = \bar\sigma + \Psi\bar u.$
Substituting this relation into the third derivative of
\eqref{eq:Robots:ref}
yields
\begin{equation*}
    \dddot{\bar z}
    =
    (\mathbb A_m\otimes I_2)
    (\bar\sigma+\Psi\bar u)
    +
    (\mathbb A_0\otimes I_2)\dddot{\bar y}_0.
\end{equation*}
Hence, the control law \eqref{eq:Robots:control} becomes
\begin{equation*}
    \bar u
    =
    \Psi^{-1}
    \bigl[
        -\bar\sigma
        +
        (\mathbb A_m\otimes I_2)
        (\bar\sigma+\Psi\bar u)
        +
        (\mathbb A_0\otimes I_2)\dddot{\bar y}_0
        +
        \mathbf K\bar\eta
    \bigr].
\end{equation*}
Collecting the terms involving $\bar u$ gives
\begin{equation}\label{eq:Robots:control_actual}
    \mathbf M\bar u
    =
    \Psi^{-1}
    \bigl[
        -\bar\sigma
        +
        (\mathbb A_m\otimes I_2)\bar\sigma
        +
        (\mathbb A_0\otimes I_2)\dddot{\bar y}_0
        +
        \mathbf K\bar\eta
    \bigr],
\end{equation}
where
\begin{equation*}
    \mathbf M
    \coloneqq
    \Psi^{-1}
    [I_{2m}-(\mathbb A_m\otimes I_2)]
    \Psi.
\end{equation*}

We now show that $\mathbf M$ is invertible.
Since the weights are normalized, $\mathbb W=I_m$,
and therefore $\mathbb L = I_m-\mathbb A_m.$ Consequently,
$I_{2m}-(\mathbb A_m\otimes I_2) = \mathbb L\otimes I_2.$ Under Remark~\ref{rem:reachability},
$\mathbb L$ is positive stable and therefore nonsingular.
Since $\Psi$ is invertible whenever $v_{x,i}\neq0$,
it follows that $\mathbf M$ is invertible along admissible trajectories.

Under the feedback law \eqref{eq:Robots:control}, the nonlinear network dynamics reduce exactly to the linear closed-loop system \eqref{eq:Robots:linear_cl}. The following result establishes nominal exponential tracking.

\begin{lemma}\label{prop:nominal_tracking}
Consider the control law \eqref{eq:Robots:control} with $\bar u^+=\mathbf K\bar\eta$,
where $\mathbf A+\mathbf B\mathbf K$ is Hurwitz. Then the origin $\bar\eta=0$
of \eqref{eq:Robots:linear_cl} is exponentially stable, i.e., there exist $c_1,c_2>0$ such that
\begin{equation}\label{eq:prop_exp_bound}
    \|\bar\eta(t)\|
    \le
    c_1e^{-c_2 t}\,\|\bar\eta(0)\|,
    \qquad \forall t\ge 0.
\end{equation}
Moreover, $\bar\epsilon(t)\to 0$ exponentially as $t\to\infty$. If, in addition, $\mathbb{L}=\mathbb{L}_m+\mathbb{A}_0$ is invertible and satisfies the balance condition \eqref{eq:ComNet:balance}, then
$\lim_{t\to\infty}\|\bar y(t)-\bar y_0(t)\|=0$.
\end{lemma}

\begin{proof}
Under
$\bar u^+=\mathbf K\bar\eta$, \eqref{eq:Robots:linear_cl} becomes $\dot{\bar\eta} = (\mathbf A+\mathbf B\mathbf K)\bar\eta.$
Since
$\mathbf A+\mathbf B\mathbf K$
is Hurwitz, standard linear systems theory implies exponential stability of the origin, yielding \eqref{eq:prop_exp_bound}. Because $\bar\epsilon$ is a subvector of $\bar\eta$, it follows immediately that
$\bar\epsilon(t)\to0$
exponentially.
Next, recall \eqref{eq:Robots:error}.
Since $\mathbb{L}$ is invertible by Remark~\ref{rem:reachability}, define $\bar y^\star \coloneqq (\mathbb{L}^{-1}\mathbb{A}_0\otimes I_2)\bar y_0$,
so that $(\mathbb{L}\otimes I_2)(\bar y-\bar y^\star)=\bar\epsilon$.
Hence,
\begin{equation*}
    \|\bar y(t)-\bar y^\star(t)\|
    \le
    \|\mathbb{L}^{-1}\otimes I_2\|\,\|\bar\epsilon(t)\|
    \to 0,
    \qquad t\to\infty.
\end{equation*}
It remains to show $\bar y^\star(t)=\bar y_0(t)$. Using \eqref{eq:ComNet:balance}, $\mathbb{L}\mathbf{1}_m=\mathbb{A}_0\mathbf{1}_m$.
Left-multiplying by $\mathbb{L}^{-1}$ yields $\mathbf{1}_m=\mathbb{L}^{-1}\mathbb{A}_0\mathbf{1}_m$.
Therefore,
\begin{equation*}
    \bar y^\star =
    (\mathbb{L}^{-1}\mathbb{A}_0\otimes I_2)(\mathbf{1}_m\otimes y_0) =
    (\mathbf{1}_m\otimes y_0) =
    \bar y_0.
\end{equation*}
Combining the above implies $\|\bar y(t)-\bar y_0(t)\|\to 0$.
\end{proof}

\section{Distributed Control Under Adversarial Interactions}\label{sec:AdversarialControl}

This section addresses the effect of adversarial interactions on the distributed nonlinear control developed in Section~4. Building on the exact linear error dynamics obtained via global input--output feedback linearization, we design a resilient desired signal at the network level and establish robustness guarantees under bounded adversarial perturbations.

\subsection{Resilient $z_i$ via Trusted Redundancy}\label{subsec:ResilientRef}

We propose a resilient construction of the local desired signal that mitigates adversarial neighbor information by combining graph redundancy with trusted navigator signals.

\begin{assumption}\label{ass:vartheta}
    For each vehicle $i$, there exists an integer $\vartheta_i\ge0$ such that
    $|\mathcal N_i^a|\le \vartheta_i$ and the in-neighbor set satisfies
    $|\mathcal N_i|\ge 2\vartheta_i+1$.
    Moreover, vehicle $i$ has access to $\vartheta_i$ trusted neighbors
    and the navigator is never corrupted, i.e., $0\notin\mathcal N_i^a$.
\end{assumption}

Let $y_i\in\mathbb R^2$ denote the planar position output of vehicle $i$.
Vehicle $i$ receives neighbor outputs $\{y_j\}_{j\in\mathcal N_i}$, where signals from agents in $\mathcal N_i^a$ may be corrupted, and recall that
$\mathcal N_i = \mathcal N_i^- \cup \mathcal N_i^a$.
Under Assumption~\ref{ass:vartheta},
\begin{equation*}
    |\mathcal N_i^-|
    =
    |\mathcal N_i|-|\mathcal N_i^a|
    \ge
    (2\vartheta_i+1)-\vartheta_i
    =
    \vartheta_i+1.
\end{equation*}
Moreover, there exists a known trusted subset
$\mathcal N_i^{\mathrm{tr}}\subseteq \mathcal N_i^-$
with
$|\mathcal N_i^{\mathrm{tr}}|=\vartheta_i$.
By assumption, if $0\in\mathcal N_i$, then
$0\in\mathcal N_i^{\mathrm{tr}}$.

For each neighbor
$k\in\mathcal N_i\setminus\mathcal N_i^{\mathrm{tr}}$,
define the deviation score
\begin{equation}\label{eq:score}
    \Delta_{ik}
    \coloneqq
    \frac{1}{|\mathcal N_i^{\mathrm{tr}}|}
    \sum\nolimits_{\ell\in\mathcal N_i^{\mathrm{tr}}}
    \|y_\ell-y_k\|.
\end{equation}
Define $\mathcal R_i$ as any index set satisfying
$|\mathcal R_i|=\vartheta_i$
and
$\Delta_{ik}\ge\Delta_{i\ell}$,
for all
$k\in\mathcal R_i$
and
$\ell\notin\mathcal R_i$.
The resilient neighbor set is then defined as
\begin{equation}\label{eq:res_set}
    \mathcal N_i^{\mathrm{res}}
    \coloneqq
    \mathcal N_i\setminus\mathcal R_i.
\end{equation}
Since $|\mathcal N_i^{\mathrm{res}}| = |\mathcal N_i| - \vartheta_i \ge \vartheta_i+1$, the set $\mathcal N_i^{\mathrm{res}}$ is nonempty.

The trimming rule \eqref{eq:res_set} removes up to $\vartheta_i$ neighbor signals with the largest deviations from the trusted set. Consequently, the resilient desired signal is constructed using only the retained neighbor information:
\begin{equation}\label{eq:z_res}
    z_i^{\mathrm{res}}
    \coloneqq
    \sum\nolimits_{j\in\mathcal N_i^{\mathrm{res}}}
    \tilde w_{ij}y_j,
    \qquad
    \tilde w_{ij}>0,
\end{equation}
where the weights satisfy $\sum\nolimits_{j\in\mathcal N_i^{\mathrm{res}}} \tilde w_{ij} = 1.$
The corresponding resilient tracking error is defined by
\begin{equation}\label{eq:error_res}
    \epsilon_i^{\mathrm{res}}
    \coloneqq
    y_i-z_i^{\mathrm{res}}.
\end{equation}

In adversarial scenarios, $\bar z$ is replaced by $\bar z^{\mathrm{res}}$ in the feedback-linearized error dynamics, yielding the same structure as \eqref{eq:Robots:error_dyn} with $\dddot{\bar z}$ replaced by $\dddot{\bar z}^{\mathrm{res}}$.
If Assumption~\ref{ass:vartheta} is violated, we revert to the worst-case robustness analysis.

\subsection{Robustness without Graph Redundancy}\label{subsec:WorstCase}

If Assumption~\ref{ass:vartheta} does not hold for some vehicles, resilient desired-signal formation is not guaranteed. In this case, adversarial interactions are treated as bounded disturbances entering the feedback-linearized error dynamics. As a result, after applying the nonlinear cancellation law \eqref{eq:Robots:control}, the closed-loop error dynamics become
\begin{equation}\label{eq:WorstCase:eta_disturbed}
    \dot{\bar\eta}
    =
    (\mathbf A+\mathbf B\mathbf K)\bar\eta
    +
    \mathbf B\bar \rho,
\end{equation}
where $\bar\rho\in\mathbb R^{2m}$ collects residual terms caused by the corrupted desired-signal derivatives. Under the bounded-derivative attack assumption, there exists $\rho_d>0$ such that
$\|\bar\rho(t)\|\le \rho_d\bar a$.

\begin{lemma}\label{prop:WorstCaseUB}
Assume $\|a_{ij}^{(k)}(t)\|\le \bar a$ for $k=0,\dots,3$ and all $(i,j)\in\mathcal E_a$, and let $\mathbf A+\mathbf B\mathbf K$ be Hurwitz.
Then, for the disturbed system \eqref{eq:WorstCase:eta_disturbed}, the output tracking error is ultimately bounded. In particular, there exists a function $\delta(\cdot)$ such that
\begin{equation}\label{eq:WorstCase:ultimate_bound}
    \limsup_{t\to\infty}\|y_i(t)-y_0(t)\|
    \le
    \delta(\bar a),
    \qquad \forall i\in\{1,\dots,m\}.
\end{equation}
\end{lemma}

\begin{proof}
Let $\mathbf A_c \coloneqq \mathbf A+\mathbf B\mathbf K$. Since $\mathbf A_c$ is Hurwitz, for any given matrix $\mathbf{Q}\succ 0$ there exists a unique $\mathbf{P}\succ 0$ solving the Lyapunov equation
$\mathbf A_c^\top \mathbf{P} + \mathbf{P}\mathbf A_c = -\mathbf{Q}$.
Consider the Lyapunov function $V(\bar\eta)\coloneqq \bar\eta^\top \mathbf{P}\bar\eta$. Along trajectories of \eqref{eq:WorstCase:eta_disturbed},
$\dot V = -\bar\eta^\top \mathbf{Q}\bar\eta + 2\bar\eta^\top \mathbf{P}\mathbf B\, \bar \rho$.
Using Cauchy--Schwarz and $\|\bar \rho(t)\|\le \rho_d\bar a$ gives
\begin{equation*}
    2\bar\eta^\top \mathbf{P}\mathbf B\, \bar\rho(t)
    \le
    2\|\mathbf{P}\mathbf B\|\,\|\bar\eta\|\,\|\bar\rho(t)\|
    \le
    2\|\mathbf{P}\mathbf B\|\,\rho_d\bar a\,\|\bar\eta\|.
\end{equation*}
Moreover, $\bar\eta^\top \mathbf{Q}\bar\eta \ge \lambda_{\min}(\mathbf{Q})\|\bar\eta\|^2$. Hence,
\begin{equation}\label{eq:WorstCase:VdotBound}
    \dot V \le
    -\lambda_{\min}(\mathbf{Q})\|\bar\eta\|^2 +
    2\|\mathbf{P}\mathbf B\|\,\rho_d\bar a\,\|\bar\eta\|.
\end{equation}
It follows that whenever $\|\bar\eta\| \ge 2\|\mathbf{P}\mathbf B\|\,\rho_d\bar a/\lambda_{\min}(\mathbf{Q})$,
we have $\dot V\le 0$, which implies ultimate boundedness of $\bar\eta(t)$ with an ultimate bound proportional to $\bar a$. In particular,
\begin{equation}\label{eq:WorstCase:etaUB}
    \limsup_{t\to\infty}\|\bar\eta(t)\|
    \le
    \frac{2\|\mathbf{P}\mathbf B\|}{\lambda_{\min}(\mathbf{Q})}\,\rho_d\bar a.
\end{equation}
Finally, note that $\bar\epsilon$ is a linear projection of $\bar\eta$. Therefore,
there exists a constant selection matrix $S\in\mathbb R^{2m\times 6m}$ such that
$\bar\epsilon=S\bar\eta$ and
$\|\bar\epsilon(t)\|\le \|S\|\,\|\bar\eta(t)\|$.
Moreover, from \eqref{eq:Robots:error} and \eqref{eq:ComNet:balance},
we have $(\mathbb L\otimes I_2)(\bar y-\bar y_0) = \bar\epsilon$.
Since $\mathbb L$ is nonsingular, $\|\mathbb L^{-1}\|$ is finite and there exists a constant $c_\epsilon>0$ such that
\begin{equation}
    \|\bar y(t)-\bar y_0(t)\|
    \le
    \|\mathbb L^{-1}\otimes I_2\|\,\|\bar\epsilon(t)\|
    \le
    c_\epsilon\|S\|\,\|\bar\eta(t)\|.
\end{equation}
Combining this with \eqref{eq:WorstCase:etaUB} yields
\eqref{eq:WorstCase:ultimate_bound} with
\begin{equation*}
    \delta(\bar a)
    \coloneqq
    c_\epsilon\|S\|\,\frac{2\|\mathbf P\mathbf B\|}{\lambda_{\min}(\mathbf Q)}\rho_d\bar a .
\end{equation*}
establishing ultimate boundedness of the tracking error.
\end{proof}

\section{Numerical Simulations}\label{sec:Sims}

We simulate the networked nonholonomic vehicles governed by \eqref{eq:Robots:nonlinear_dynamics}. The controller follows \eqref{eq:Robots:error_dyn}--\eqref{eq:Robots:linear_cl}, where $\bar u$ is implemented via \eqref{eq:Robots:control_actual}. We consider $m=12$ tracking vehicles and three directed topologies: the star $\mathcal{G}_{\mathrm{s}}$, cyclic $\mathcal{G}_{\mathrm{c}}$, and path $\mathcal{G}_{\mathrm{p}}$ graphs shown in Fig.~\ref{Fig:topology}. The simulations are performed over the interval $t\in[0,30]$ using MATLAB \texttt{ode45} with relative and absolute tolerances $10^{-7}$ and $10^{-9}$, respectively. The feedback gain is chosen as $\mathcal{K} = -\diag\{(1,3,3),(1,3,3)\}$ with $\mathbf K = I_m\otimes \mathcal K$. Initial conditions satisfy $(p_{x,i}(0),p_{y,i}(0))\in[-200,200]^2$, $\theta_i(0)\in[-\pi,\pi]$, $v_{x,i}(0)\in[0.2,0.4]$, and $\omega_i(0)=F_i(0)=0$. The navigator trajectory is generated analytically using smooth sinusoidal reference signals. 
For evaluation, we report the average leader-tracking error
$\tilde e(t)\coloneqq \frac{1}{m}\sum_{i=1}^m \|y_i-y_0\|$
and the average disagreement error
$\tilde \epsilon(t)\coloneqq \frac{1}{m}\sum_{i=1}^m \|y_i-z_i\|$.

\textbf{Without Adversarial Interactions}: Figures~\ref{fig:star_1}--\ref{fig:path_1} show that for all three topologies the tracking vehicles converge to the navigator trajectory, i.e., $\bar e(t)\to 0$ (Fig.~\ref{fig:Err_1}, top).
This topology-independence is consistent with the feedback-linearized error dynamics \eqref{eq:Robots:linear_cl} and the choice of $\mathbf{K}$ such that $\mathbf{A}+\mathbf{B}\mathbf{K}$ is Hurwitz.
Figure~\ref{fig:Err_1} also shows that $\tilde e(t)$ and $\tilde \epsilon(t)$ need not coincide.
Indeed, by definition the term $z_i=\sum_{j\in\mathcal{N}_i}w_{ij}y_j$ in \eqref{eq:Problem:loc_ref}, so $z_i$ is a \emph{topology-dependent} convex combination of neighbor outputs whereas $y_0$ is the same exogenous reference for all vehicles. Hence,
\begin{equation*}
    \|y_i-y_0\| \le
    \|y_i-z_i\|+\|z_i-y_0\|,
\end{equation*}
and the term $\|z_i-y_0\|$ depends on the local mixing induced by $(\mathbb{A}_m,\mathbb{A}_0)$ even when all vehicles ultimately satisfy $y_i\to y_0$.
This mixing effect causes $\tilde \epsilon(t)$ to differ across topologies, while $\tilde e(t)$ remains nearly identical (Fig.~\ref{fig:Err_1}).
\begin{figure}[t!]
    \centering
    \scalebox{1}{{\begin{tikzpicture}
            \centering
            \Text[x=0,y=1.5]{$\mathcal{G}_{\mathrm{s}}\coloneqq$ Star};
            \Vertex[x=0,y=0,label=$0$,color=blue,opacity=0.2,size=.5]{L}
            \Vertex[x=1,y=1,label=$1$,color=red,opacity=0.2,size=.5]{1}
            \Vertex[x=-1,y=1,label=$2$,color=red,opacity=0.2,size=.5]{2}
            \Vertex[x=-1,y=-1,label=$3$,color=red,opacity=0.2,size=.5]{3}
            \Vertex[x=1,y=-1,label=$4$,color=red,opacity=0.2,size=.5]{4}
            \Edge[Direct,color=blue,label=$1.0$](L)(1)
            \Edge[Direct,color=blue,label=$1.0$](L)(2)
            \Edge[Direct,color=blue,label=$1.0$](L)(3)
            \Edge[Direct,color=blue,label=$1.0$](L)(4)
        \end{tikzpicture}}}
    \qquad
    \scalebox{1}{{\begin{tikzpicture}
            \centering
            \Text[x=0,y=1.5]{$\mathcal{G}_{\mathrm{c}}\coloneqq$ Cyclic};
            \Vertex[x=0,y=0,label=$0$,color=blue,opacity=0.2,size=.5]{L}
            \Vertex[x=1,y=1,label=$1$,color=red,opacity=0.2,size=.5]{1}
            \Vertex[x=-1,y=1,label=$2$,color=red,opacity=0.2,size=.5]{2}
            \Vertex[x=-1,y=-1,label=$3$,color=red,opacity=0.2,size=.5]{3}
            \Vertex[x=1,y=-1,label=$4$,color=red,opacity=0.2,size=.5]{4}
            \Edge[Direct,color=blue,label=$0.1$](L)(1)
            \Edge[Direct,color=blue,label=$0.1$](L)(2)
            \Edge[Direct,color=blue,label=$0.1$](L)(3)
            \Edge[Direct,color=blue,label=$0.1$](L)(4)
            \draw[<->,ultra thick, latex' -latex'] (1) -- node[midway, fill=white, inner sep=2pt]{\small $0.45$} (2);
            \draw[<->,ultra thick, latex' -latex'] (2) -- node[midway, fill=white, inner sep=2pt]{\small $0.45$} (3);
            \draw[<->,ultra thick, latex' -latex'] (3) -- node[midway, fill=white, inner sep=2pt]{\small $0.45$} (4);
            \draw[<->,ultra thick, latex' -latex'] (4) -- node[midway, fill=white, inner sep=2pt]{\small $0.45$} (1);
        \end{tikzpicture}}}
    \qquad
    \scalebox{1}{{\begin{tikzpicture}
            \centering
            \Text[x=0,y=1.5]{$\mathcal{G}_{\mathrm{p}}\coloneqq$ Path};
            \Vertex[x=0,y=0,label=$0$,color=blue,opacity=0.2,size=.5]{L}
            \Vertex[x=1,y=1,label=$1$,color=red,opacity=0.2,size=.5]{1}
            \Vertex[x=-1,y=1,label=$2$,color=red,opacity=0.2,size=.5]{2}
            \Vertex[x=-1,y=-1,label=$3$,color=red,opacity=0.2,size=.5]{3}
            \Vertex[x=1,y=-1,label=$4$,color=red,opacity=0.2,size=.5]{4}
            \Edge[Direct,color=blue,label=$1.0$](L)(1)
            \draw[<->,ultra thick, -latex'] (1) -- node[midway, fill=white, inner sep=2pt]{\small $1.0$} (2);
            \draw[<->,ultra thick, -latex'] (2) -- node[midway, fill=white, inner sep=2pt]{\small $1.0$} (3);
            \draw[<->,ultra thick, -latex'] (3) -- node[midway, fill=white, inner sep=2pt]{\small $1.0$} (4);
        \end{tikzpicture}}}
    \caption{Three network topologies $(\mathcal{G}_{\mathrm{s}}, \mathcal{G}_{\mathrm{c}}, \mathcal{G}_{\mathrm{p}})$ with weights used in the simulations.}
    \label{Fig:topology}
\end{figure}
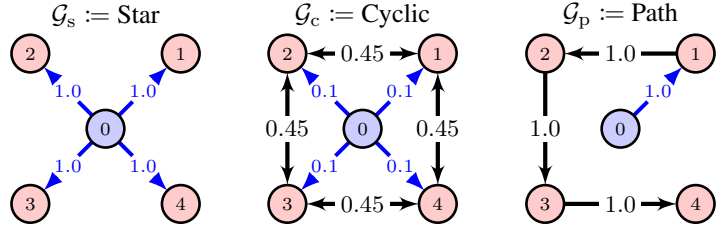

\textbf{With Adversarial Interactions}:
We consider adversarial interactions with the same attack magnitude across the three topologies. In the star topology, adversarial corruption enters through the navigator channel: for a fixed attacked set (e.g., $\{2,5,8,11\}$), the received navigator signal $y_0$ is corrupted. In the cyclic and path topologies, adversarial effects enter through inter-vehicle communication: when $j$ belongs to the attacked set, the transmitted neighbor output $y_j$ is corrupted. In all cases, the resilient desired signal is constructed according to Subsection~\ref{subsec:ResilientRef}; accordingly, $z_i$ is replaced by $z_i^{\mathrm{res}}$ in the controller implementation.

Figures~\ref{fig:star_2}--\ref{fig:path_2} show that the cyclic topology exhibits the strongest robustness under adversarial interactions.
Structurally, each agent $i$ in the cyclic graph has at least $2\vartheta+1$ in-neighbors and direct access to the trusted navigator $y_0$, allowing the trimming rule to discard corrupted neighbor signals while preserving a consistent reference anchored at $y_0$. By contrast, the star topology is vulnerable when the navigator channel is corrupted for some agents, since the trusted-navigator assumption is violated at the measurement level and trimming cannot correct a corrupted $y_0$.
The path topology is also fragile due to its single-chain information flow: once an upstream transmission is corrupted, downstream agents lack sufficient redundancy to reject it.

These effects are summarized in Fig.~\ref{fig:Err_2}.
The cyclic topology achieves the smallest average leader-tracking error (top subplot) and the lowest disagreement to the network reference (bottom subplot), whereas the star and path topologies exhibit larger residual errors due to corrupted navigator reception and limited graph redundancy, respectively.

\begin{figure*}[t!]
    \centering
    \subfloat[\label{fig:star_1} Star--$\mathcal{G}_{\mathrm{s}}$]{\includegraphics[width=.25\linewidth]{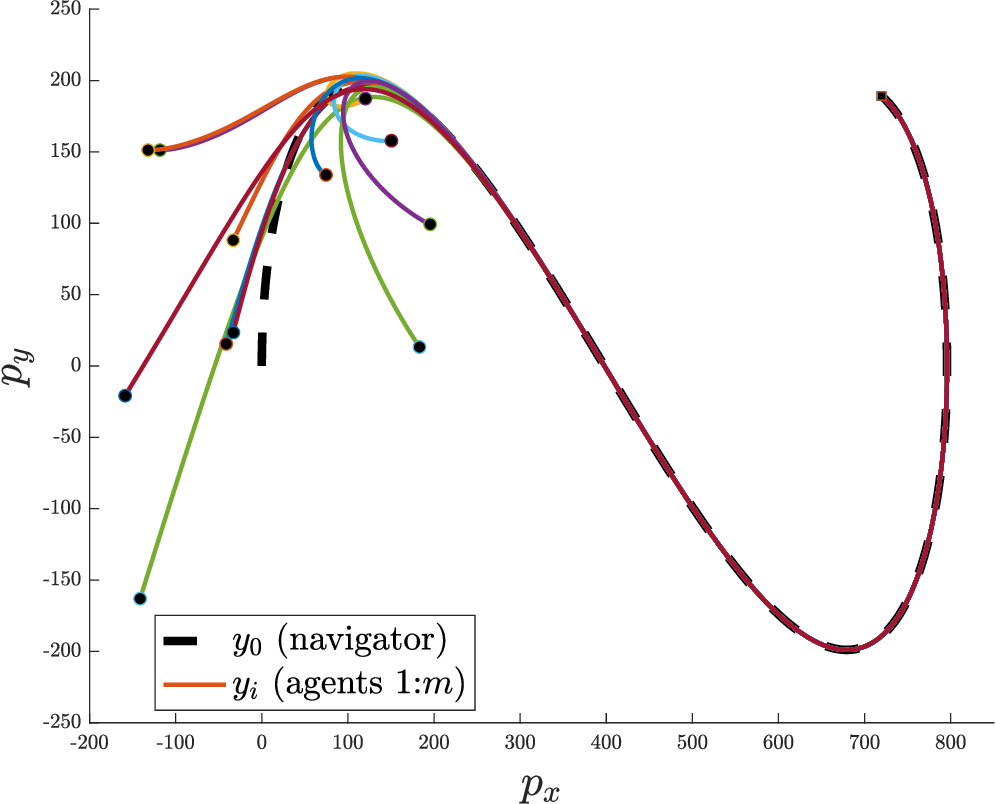}}
    \subfloat[\label{fig:cyclic_1} Cyclic--$\mathcal{G}_{\mathrm{c}}$]{\includegraphics[width=.25\linewidth]{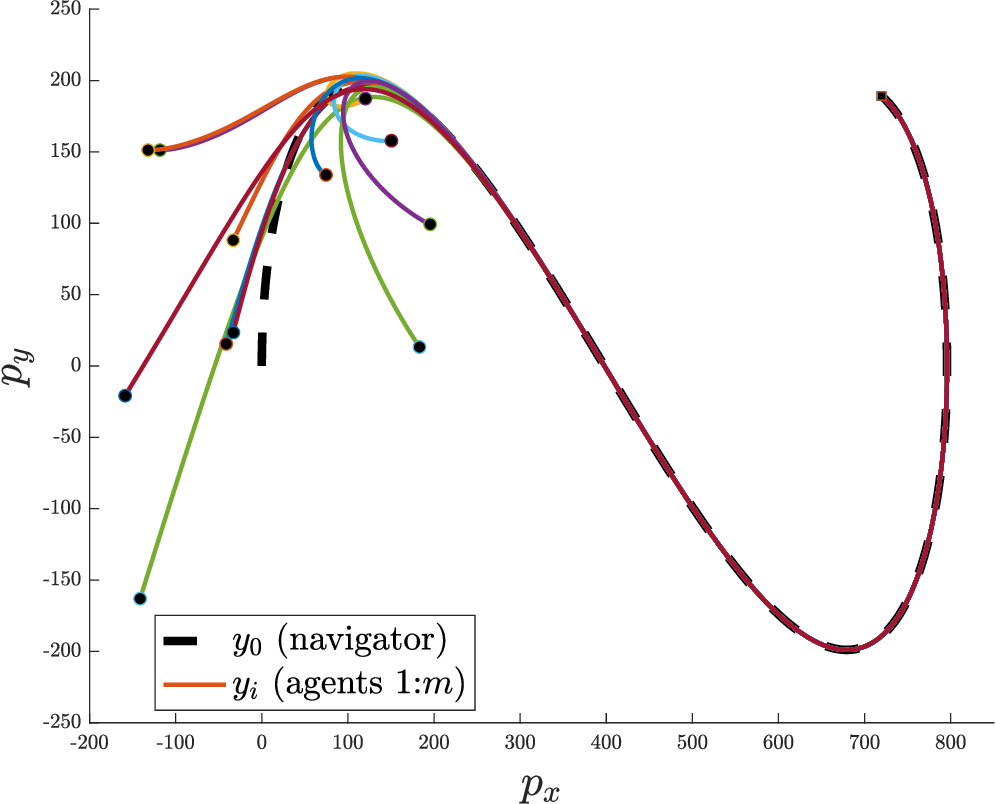}}
    \subfloat[\label{fig:path_1} Path--$\mathcal{G}_{\mathrm{p}}$]{\includegraphics[width=.25\linewidth]{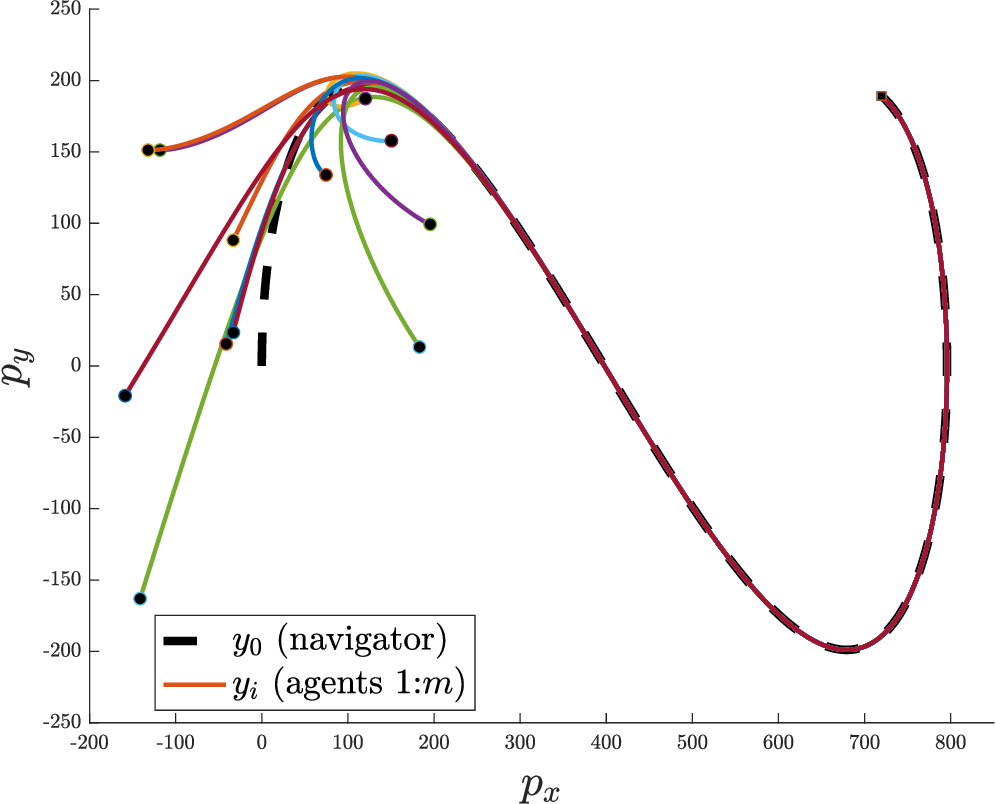}}
    \subfloat[\label{fig:Err_1} Error]{\includegraphics[width=.25\linewidth]{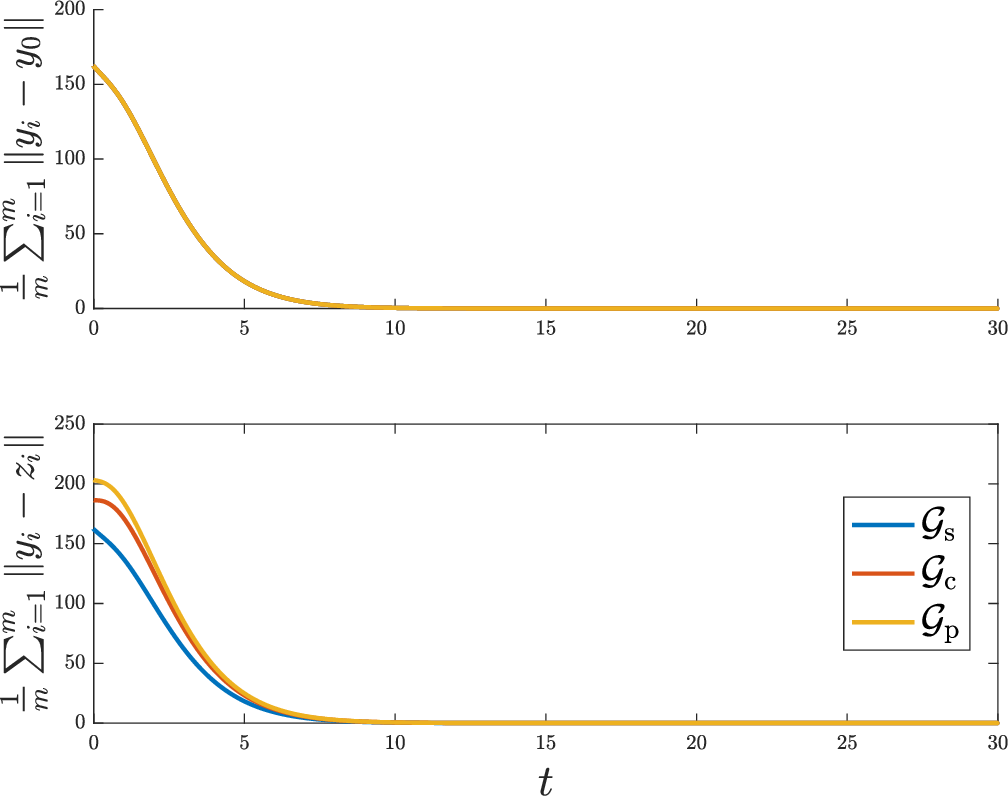}} \\
    \subfloat[\label{fig:star_2} Star--$\mathcal{G}_{\mathrm{s}}$]{\includegraphics[width=.25\linewidth]{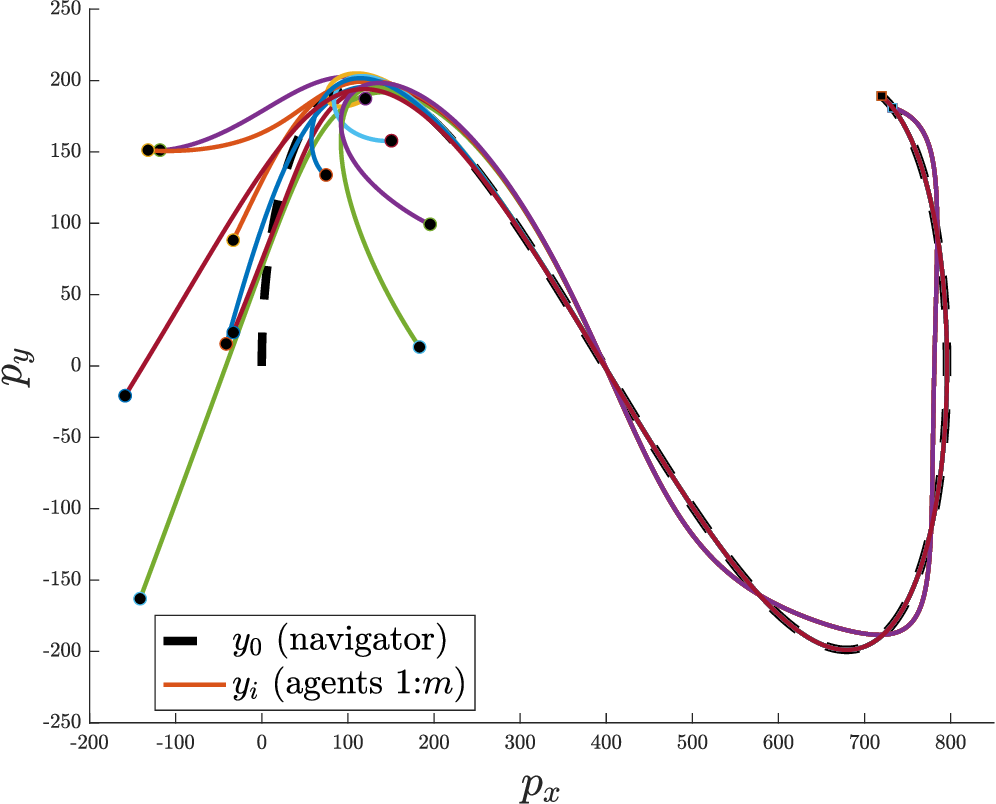}}
    \subfloat[\label{fig:cyclic_2} Cyclic--$\mathcal{G}_{\mathrm{c}}$]{\includegraphics[width=.25\linewidth]{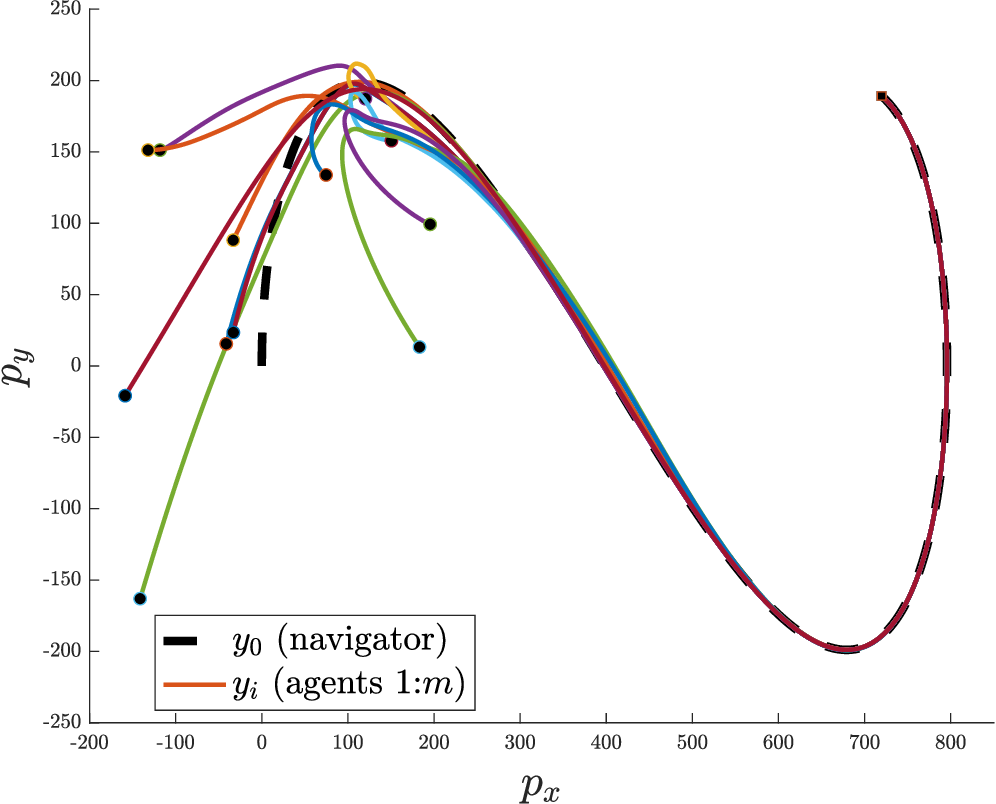}}
    \subfloat[\label{fig:path_2} Path--$\mathcal{G}_{\mathrm{p}}$]{\includegraphics[width=.25\linewidth]{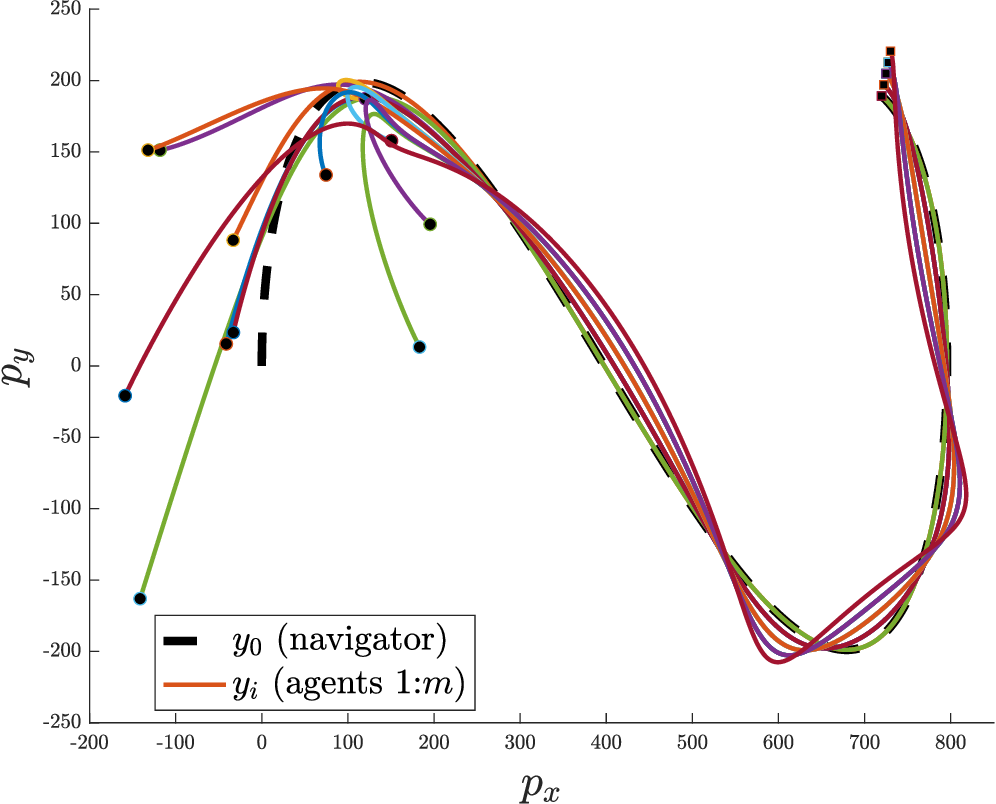}}
    \subfloat[\label{fig:Err_2} Error]{\includegraphics[width=.25\linewidth]{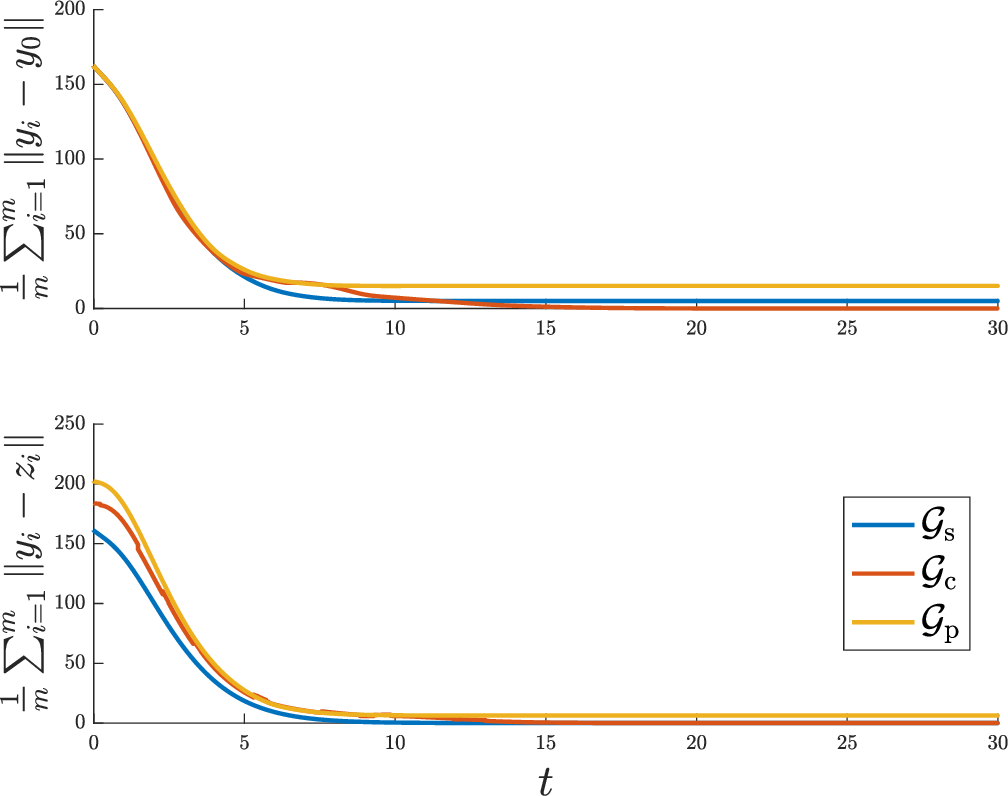}} 
    \caption{Networked trajectory tracking for three directed topologies. Top row (non-adversarial case): vehicle trajectories under star ($\mathcal{G}_{\mathrm{s}}$), cyclic ($\mathcal{G}_{\mathrm{c}}$), and path ($\mathcal{G}_{\mathrm{p}}$) graphs, and the corresponding averaged tracking and disagreement errors. Bottom row (adversarial case): trajectories and errors under the same topologies with corrupted information exchange. The error plots report the average leader-tracking error $\tilde e(t)=\frac{1}{m}\sum_{i=1}^m\|y_i-y_0\|$ and the average disagreement $\tilde\epsilon(t)=\frac{1}{m}\sum_{i=1}^m\|y_i-z_i\|$.}
    \label{fig:topologies}
\end{figure*}

\section{Conclusion}

This paper addressed distributed trajectory tracking of nonholonomic vehicle networks under adversarial information exchange. A global input--output feedback linearization was developed to regulate planar positions with exact linear error dynamics. To mitigate corrupted neighbor information, a resilient desired-signal construction was proposed by exploiting local redundancy together with trusted neighbor signals.
When each vehicle has at least $2\vartheta+1$ in-neighbors and has access to $\vartheta$ trusted signals, the proposed scheme suppresses adversarial effects and recovers nominal tracking performance. If redundancy is violated, the tracking error remains ultimately bounded under a worst-case disturbance characterization. Simulations on star, cyclic, and path topologies demonstrate that cyclic networks achieve superior robustness due to distributed information propagation. Future work will explore adaptive trust mechanisms within feedback-dependent control schemes \cite{Wafi-LCSS24,ref19,Wafi-SIAM,Wafi-MRAC}, develop estimation methods to detect and identify adversarial signals \cite{ref20,Wafi-AIP,Wafi-Elham}, and extend the framework to more general and dynamic attack models beyond the current worst-case disturbance characterization \cite{ref18,ref21}.

\bibliographystyle{ieeetr}
\bibliography{reference.bib}

\end{document}